\documentclass[prl,twocolumn, aps,preprintnumbers,showpacs, superscriptaddress]{revtex4}

\def\venue{arxiv}
\newcommand{\omitted}[1]{\ifthenelse{\equal{\venue}{arxiv}}{#1}{The proof can be found in
the extended version of this paper~\cite{BPT09}.}}

\usepackage{latexsym}
\usepackage{amsmath}
\usepackage{amssymb}
\usepackage{amsfonts}
\usepackage{ifthen}
\usepackage{revsymb}
\usepackage{yfonts}
\usepackage{graphicx}

\usepackage{graphicx}
\usepackage{amsmath}
\usepackage{amssymb}
\usepackage{amsthm}

\newcommand{\be}{\begin{equation}}
\newcommand{\ee}{\end{equation}}
\newcommand{\ba}{\begin{array}}
\newcommand{\ea}{\end{array}}
\newcommand{\bea}{\begin{eqnarray}}
\newcommand{\eea}{\end{eqnarray}}

\newcommand{\calC}{{\cal C }}
\newcommand{\calH}{{\cal H }}

\newcommand{\calE}{{\cal E }}
\newcommand{\calT}{{\cal T }}

\newcommand{\calS}{{\cal S }}

\newcommand{\calR}{{\cal R }}

\newcommand{\ZZ}{\mathbb{Z}}

\newcommand{\la}{\langle}
\newcommand{\ra}{\rangle}

\newcommand{\nn}{\nonumber}

\newcommand{\trace}{\mathop{\mathrm{Tr}}\nolimits}

\newtheorem{dfn}{Definition}
\newtheorem{lemma}{Lemma}
\newtheorem{prop}{Proposition}

\newtheorem{corol}{Corollary}
\newtheorem{fact}{Fact}

\begin{document}

\title{Tradeoffs for reliable quantum information storage in 2D systems}

\author{Sergey \surname{Bravyi}}
\affiliation{IBM Watson Research
Center,  Yorktown Heights  NY 10598, USA}
\author{David  \surname{Poulin}}
\affiliation{D\'epartement de Physique, Universit\'e de  Sherbrooke, Qu\'ebec, Canada}
\author{Barbara \surname{Terhal}}
\affiliation{IBM Watson Research
Center,  Yorktown Heights  NY 10598, USA}

\date{\today}

\begin{abstract}
We ask whether there are fundamental limits on storing quantum information
reliably in a bounded volume of space. To investigate this question,
we study quantum error correcting codes specified by geometrically
local commuting constraints on a 2D lattice of finite-dimensional
quantum particles. For these 2D systems, we derive a tradeoff
between the number of encoded qubits $k$, the distance of the code
$d$, and the number of particles $n$. It is shown that $kd^2=O(n)$
where the coefficient in $O(n)$ depends only on the locality of the
constraints and dimension of the Hilbert spaces describing
individual particles. We show that the analogous tradeoff for the
classical information storage is $k\sqrt{d} =O(n)$.
\end{abstract}


\maketitle

Understanding the limits imposed on information processing  by the
laws of physics is a problem of fundamental and practical
importance. A variety of hardware-independent limitations on the
power of computers arising from thermodynamics, quantum mechanics,
and relativity  have been identified
recently~\cite{Landauer88,Bennett85,Lloyd00}.

In this Letter we derive a fundamental upper bound on the amount of
quantum information that can be stored reliably in a given volume of
a 2D space. This bound stems from {\em geometric locality} of
quantum operations used to detect and correct errors as well as
peculiar features of {\em quantum entanglement} in 2D systems.
 We shall model the information storage using the framework of quantum error correcting codes~\cite{NielsenChuaung2000}.
Specifically, we consider a system of $n$ finite-dimensional quantum particles (qudits)
occupying sites of a 2D lattice $\Lambda$.
For the sake of clarity we shall consider a regular square lattice of size $\sqrt{n}\times \sqrt{n}$
with open boundary conditions, although  our results can be easily extended to more general
2D lattices and periodic boundary conditions.
 We  shall focus on codes for which the codespace $\calC$ spanned by encoded states can be
represented as a common eigenspace of geometrically local pairwise commuting~\footnote{The commutativity
of the constraints specifying a code is a highly desirable property since it allows one to use error correction
algorithms based on the syndrom measurement.}
projectors $\Pi_1,\ldots,\Pi_m$ such that
\be
\label{C}
\calC=\{ |\psi\ra\, : \, \Pi_a \, |\psi\ra =|\psi\ra \quad \mbox{for all $a$}\}.
\ee
The codespace $\calC$ can be regarded
as the ground-state subspace of a local gapped Hamiltonian
\be
\label{Hcode}
H=-\sum_{a=1}^m \Pi_a, \quad \Pi_a \Pi_b=\Pi_b \Pi_a.
\ee
Such a code is able to encode $k=\log_2 \dim{\calC}$ logical qubits.
Let $d$ be the distance of the code~\footnote{Recall that a code has
distance $d$ iff any measurement performed on any subset of less
than $d$ particles reveals no information about an encoded state.
Such a code protects encoded information against $\lfloor
(d-1)/2\rfloor$ single-particle errors in the worst-case scenario.}.
Our main result in an upper bound \be \label{quantum} k\le \frac{c\,
n}{d^2}. \ee Here $c$ is a constant coefficient that depends only on
locality of the projectors defining the codespace and dimension of
the Hilbert space describing individual particles.
 The bound Eq.~(\ref{quantum}) is tight
up to a constant factor since
2D surface codes~\cite{BK98} achieve the scaling  $kd^2 \sim n$
for any given $n$ and $d$~\footnote{Note that a 2D surface code
defined on a lattice with two smooth and two rough boundary regions
encodes one qubit into approximately $d^2$ physical qubits, see~\cite{BK98}.  In order to encode
$k$ qubits one can use $k$ independent copies of the surface code placed next to each other
on a plane. It requires roughly $kd^2$ physical qubits.}.
The bound Eq.~(\ref{quantum}) can be put in sharp contrast with the
existence of good stabilizer codes~\cite{CS:good} for which $k/n
\geq c_1$ and $d/n \geq c_2$ for some constants $c_1,c_2$. Our
result implies that the distance of 2D quantum codes with a non-zero
rate $k/n$ is upper bounded by a constant independent of $n$.
 It also implies that the distance of any 2D quantum code is at most $O(\sqrt{n})$
 extending the results of~\cite{BT09} beyond stabilizer codes.


The motivation for our work stems from several sources. Firstly,
quantum error correcting codes provide toy models for how
topological quantum order (TQO) can emerge in the ground states of
2D spin systems with short-range interactions. For example,
string-net models introduced by Levin and Wen~\cite{LevinWen03} are
described by Hamiltonians involving a sum of commuting projectors,
see~\cite{KRV08}. The ground state of string-net models defined on a
torus (or higher genus surface) has topological degeneracy and can
be regarded as a codespace of a quantum code. Alternatively, the
codespace can be chosen as an excited subspace corresponding to a
particular configuration of excitations (anyons) ---  the approach
adopted by Kitaev in the topological quantum computing
scheme~\cite{Kitaev97}. In this case the code distance is
proportional to the distance between anyons while the bound
Eq.~(\ref{quantum}) asserts  that the number of encoded qubits is at
most a constant fraction of the number of anyons.

Secondly, one can interpret Eq.~(\ref{quantum}) as a tradeoff
between {\em degeneracy} and {\em stability} that must be obeyed by
ground states of the code Hamiltonian $H$. Assuming that $H$ is
translation-invariant, one has a stable zero temperature phase in
the thermodynamic limit if the degeneracy of the ground state cannot
be lifted by weak local perturbations  below some critical value of
the perturbation parameter. It is well known that adding a weak
local perturbation to $H$ lifts the degeneracy of the ground state
only in order $\Omega(d)$ of perturbation theory~\cite{Kitaev97}.
Thus a necessary condition for $T=0$ stability is that the distance
$d$ must be  infinite in the thermodynamic limit. Then the  tradeoff
Eq.~(\ref{quantum}) implies that the {\em amount of quantum
information stored per unit volume}, $k/n \leq c/d^2$ goes to zero
in the thermodynamic limit. This suggests a possible connection
between our results and the celebrated {\em holographic principle}
asserting that the amount of  information that can be encoded in a
volume of space $M$ scales as the area of the boundary of $M$.

Generalizing our techniques to quantum codes defined on a
$D$-dimensional lattice  yields \be \label{quantum'} k\le
\frac{cn}{d^{\alpha}}, \quad \alpha = \frac2{D-1}. \ee As was shown
in Ref.~\cite{BT09}, the distance of any $D$-dimensional stabilizer
code satisfies the bound $d\leq O(n^{(D-1)/D})$. Since the upper bound
Eq.~(\ref{quantum'}) permits codes with $k=O(1)$ and  $d\sim
n^{1/\alpha} \sim n^{(D-1)/2}$, it cannot be tight for all values of
$n$ and $d$ unless $D=2$. Using the folded surface code
construction~\cite{Kitaev09} one can construct a $D$-dimensional
stabilizer code encoding $1$ qubit into $n$ qubits with the distance
$d\sim \sqrt{n}$. To the best of our knowledge there are no examples
of $D$-dimensional codes for which the distance grows faster than
$\sqrt{n}$. Therefore one cannot exclude the possibility that the
bound Eq.~(\ref{quantum}) holds for any spatial dimension, although
we consider this to be unlikely.

  It should be emphasized that
throughout this paper the  geometric locality of the constraints
$\Pi_a$ is defined using the standard Euclidean
geometry~\footnote{More strictly, our analysis applies to the
regular $D$-dimensional cubic lattice with open or periodic boundary
conditions.}. At the same time, the bound Eq.~(\ref{quantum}) can be
violated for non-Euclidean geometry. For example,
Ref.~\cite{zemor:surface} constructed surface codes on general
planar graphs with a constant rate $k/n$ and the distance $d\sim
\log{n}$, see also~\cite{BMD:rate}. Also, it is known that
stabilizer codes with $k=1$ and $d\sim \sqrt{n} \log{n}$ can be
constructed on triangulations of some 4D Riemannian surfaces, see
Theorem~12.4 in Ref.~\cite{FML:systole}.



%


We note that even though our results cover a large family of 2D
quantum codes on qudits beyond the standard family of stabilizer
codes, they do not include the important family of quantum
subsystem codes \cite{Pou05b, Bac05a}.

One can also ask about the analogue of the tradeoff
Eq.~(\ref{quantum}) for classical information storage. In Appendix~A
we prove that any 2D classical code specified by geometrically local
constraints obeys the bound
 \be \label{classical} k\le \frac{c \, n}{\sqrt{d}} \ee
Here $c$ is a constant depending only on the dimension of individual particles
and locality of the constraints specifying the code.
Using the mapping from 1D cellular automatons to 2D classical codes
we construct a family of codes with $k\sim \sqrt{n}$
and $d\sim n^{0.8}$ which is quite close to saturating the bound Eq.~(\ref{classical}).

{\bf Definitions and notations.}
We shall assume that the locality of the projectors $\Pi_a$ can be
characterized by a constant interaction range $w$ such that the
support of any projector $\Pi_a$ can be covered by a square block of
size $w\times w$. Let
 \be
  \label{PI} \Pi=\prod_{a=1}^m \Pi_a
 \ee
be the projector on the codespace $\calC$.
  A state $\rho$ is called an {\em encoded state} iff it
has support on the codespace $\calC$, that is, $\Pi\rho=\rho\Pi=\rho$.
We shall say that
a region $M\subseteq \Lambda$ is {\em correctable} iff
there exists an error correction operation (a
trace preserving completely positive map) $\calR$ that corrects the erasure
of all particles in $M$, that is,  for any encoded state $\rho$ one has
\be \label{EC} \calR(\trace_M \rho )=\rho.
\ee
By definition of the distance any region of size smaller  than $d$ is
correctable.

We shall use the notation $\bar{M}=\Lambda\backslash M$ for the
complement of a region $M$.  For any region $M\subseteq
\Lambda$ and for any fixed state $\rho$ let $S(M)=-\trace \rho_M
\log \rho_M$ be the von Neumann entropy of the reduced density
matrix $\rho_M$.
Using techniques from Ref.~\cite{SN96} one can easily show
 that the error correction condition Eq.~(\ref{EC})
has the following entropic counterpart.
\begin{fact}
If a region $M$ is correctable then \be \label{EC2}
S(M|\bar{M})=-S(M) \ee for any encoded state $\rho$. Here
$S(M|\bar{M})=S(M\bar{M})-S(\bar{M})$ is the entropy
of $M$ conditioned on $\bar{M}$. \label{fact:correct}
\end{fact}

Note that the equality Eq.~(\ref{EC2}) holds automatically for any pure state of $M\bar{M}$
which would correspond to a trivial code with $k=0$. More generally, Eq.~(\ref{EC2}) implies  that
there exists a (virtual) partition  $\bar{M}=AB$ such that any encoded state $\rho$ is a tensor product of
some fixed pure state held by $MA$ and some state of $B$ depending on $\rho$ \cite{HJPW03}.

\omitted{
\begin{proof}
Let $\rho_{M\bar{M}}$ be any encoded state and
$\rho_{M\bar{M}C}$ be its purification. Define an error
$\calT=\trace_M\otimes \mathrm{id}_{\bar{M}C}$ erasing the
region $M$. By assumption there exists a recovery operation $\calR$
such that
\[
\calR\circ \calT (\rho_{M\bar{M} C}) = \rho_{M\bar{M} C},
\quad \calR\circ \calT (\rho_{M\bar{M}}\otimes \rho_C) =
\rho_{M\bar{M}}\otimes \rho_C.
\]
Therefore \be \label{rel_ent}
\calS(\rho_{M\bar{M}C}||\rho_{M\bar{M}}\otimes \rho_C) =
\calS(\calT(\rho_{M\bar{M}C})||\calT(\rho_{M\bar{M}}\otimes
\rho_C)) \ee since the relative entropy is monotone decreasing under
quantum operations, see~\cite{Uhlmann77}. Using the definition of
$\calT$ one can rewrite Eq.~(\ref{rel_ent}) as \be\label{EC2'}
\calS(\rho_{M\bar{M}C}||\rho_{M\bar{M}}\otimes
\rho_C)=S(\rho_{\bar{M}C}||\rho_{\bar{M}}\otimes \rho_C).
\ee Taking into account that $\rho_{M\bar{M}C}$ is a pure
state, one can check that  Eq.~(\ref{EC2'}) is equivalent to Eq.~(\ref{EC2}).
\end{proof}
}


We begin by sketching the steps leading up to our main result,
the bound in Eq.~(\ref{quantum}). Let $R$ be the largest integer
$m$ such that any square block of size $m\times m$ is correctable.
Note that $R$ is at least $\sqrt{d}$ by the definition of the
distance.

 Consider a partition of the lattice $\Lambda=ABC$ shown in
Fig.~\ref{fig:ABC}. The regions $A$ and $B$ consist of blocks of
size $R\times R$, so that each individual block in $A$ and $B$ is
correctable.
The total number of blocks is roughly $n/R^2$.
The regions $A$ and $B$ have small corner regions taken
out which make up the region $C$.
The purpose of the region $C$ is
to provide a sufficiently large separation between the neighboring
blocks in $A$ and between the neighboring blocks in $B$ such that
any projector $\Pi_a$ overlaps with at most one block in $A$ and
with at most one block in $B$. It guarantees that the entire regions
$A$ and $B$ are correctable (see Lemma~\ref{lemma:union} below).
Applying Eq.~(\ref{EC2}) to regions $A$ and $B$ yields \be
\label{cond_ent} S(A|BC)=-S(A) \quad \mbox{and} \quad S(B|AC)=-S(B)
\ee for any encoded state. Let $\rho$ be the maximally mixed encoded
state such that $k=S(\Lambda)$. Using Eq.~(\ref{cond_ent}) we get
\bea \label{term1} S(\Lambda)&=&S(BC)+S(A|BC)
= S(BC) - S(A) \nn \\
& \le &  S(C)+ S(B)-S(A).
\eea
Similarly
\bea
\label{term2}
S(\Lambda)&=& S(AC)+S(B|AC) = S(AC) - S(B) \nn \\
&\le&  S(C) + S(A) - S(B). \eea Adding together
Eqs.~(\ref{term1},\ref{term2}) yields \be \label{k_upper}
k=S(\Lambda) \le S(C) \le |C| \sim
\frac{n}{R^2}. \ee The second step in the proof which may be
less intuitive is to show that $R\ge cd$ for some constant $c$
depending only on locality of the constraints. In other words, we
need to prove that any block of size roughly $d\times d$ is
correctable. Our main technical tool will be the {\em Disentangling
Lemma} characterizing entanglement properties of the maximally mixed
encoded state proportional to the projector on the codespace $\Pi$.
We shall prove that any correctable region $M$ can be completely
disentangled from the rest of the lattice by acting only on the
boundary of the region (see Lemma~\ref{lemma:disent} below). The
disentangling operation leaves the region $M$ in a pure state, so
that all entropy of $M$ can be ``cleaned out" by acting along the
boundary of $M$. This result can be regarded  as a generalization of
the Cleaning Lemma from~\cite{BT09} beyond stabilizer codes.
For any region $M$ let $\partial M$ be the boundary of $M$, that is,
the region covered by the supports of all projectors $\Pi_a$ that
couple $M$ with $\bar{M}$. The following result is a simple
corollary of the Disentangling Lemma.
\begin{corol}
\label{cor:recursion} Let $M$ be any correctable region. Consider
any regions $B\subseteq M$ and $C\subseteq \bar{M}$ such that
$BC$ is correctable  and $\partial M \subseteq BC$. Then $M\cup C$
is also correctable.
\end{corol}
The idea of the proof is illustrated in Fig.~\ref{fig:recursion}.
Let us apply Corollary~\ref{cor:recursion} to a square block $M$ of
size $R\times R$. Choose $B$ and $C$ as layers of thickness
$w$ adjacent to the surface of $M$ such that $B\subseteq M$ and
$C\subseteq \bar{M}$, see Fig.~\ref{fig:recursion}. Since all
the projectors $\Pi_a$ have  size at most $w$,  the
condition $\partial M \subseteq BC$ is satisifed. Note that
$|BC|= cwR$ for some constant $c$. If $|BC|<d$ then $BC$ is
correctable and Corollary~\ref{cor:recursion} would imply that
$M\cup C$ is correctable. But $M\cup C$ is a square block of size
larger than $R$ which contradicts the choice of $R$. Thus
$|BC|\ge d$, that is, $R\ge d/(cw) \sim d$. Substituting this
bound into Eq.~(\ref{k_upper}) completes the proof of
Eq.~(\ref{quantum}).


Let us comment on how to extend this proof technique to $D$-dimensional
lattices. The partition $\Lambda=ABC$ of Fig.~\ref{fig:ABC}
should be chosen such that $A$ and $B$ consist of $D$-dimensional
cubes of linear size $R$.  Adjacent cubes in $A$ or $B$ overlap along
$(D-2)$-dimensional faces. Accordingly, the region $C$ is a union of all
$(D-2)$-dimensional faces (with a thickness of order $w$) over all blocks in $A$ and $B$.
Note that $|C|\sim n/R^2$, so we arrive at Eq.~(\ref{k_upper}).
Repeating the same arguments as above shows that a cubic-shaped
block $M\subseteq \Lambda$ is correctable if $|\partial M|<d$, that is,
$R^{D-1}\sim d$. Substituting it into Eq.~(\ref{k_upper}) leads to Eq.~(\ref{quantum'}).


\begin{center}
\begin{figure}[htb]
\centerline{
\mbox{
 \includegraphics[height=3cm]{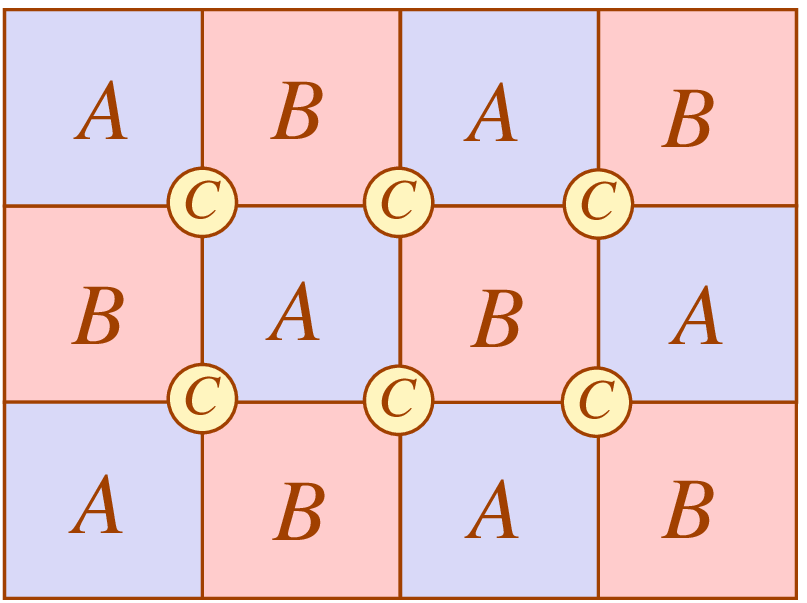}
 }}
\caption{The partition of the lattice $\Lambda=ABC$.
Each individual block in $A$ and $B$ must be correctable.
The region $C$ provides separation between adjacent blocks in $A$ and adjacent blocks in $B$.
It guarantees that the entire regions $A$ and $B$ are correctable.
The entropic error correction condition implies that
 $S(A|BC)=-S(A)$ and $S(B|AC)=-S(B)$ for the maximally mixed encoded state.
It yields $k=S(ABC)\le S(C)$. }
 \label{fig:ABC}
\end{figure}
\end{center}

In the rest of the paper we state and prove the Disentangling Lemma,
provide a formal proof of Corollary~\ref{cor:recursion}, and prove that a union 
of correctable sets that are sufficiently far from each other is also a correctable set,
see Lemma~\ref{lemma:union}.
\begin{dfn}
Let $M\subseteq \Lambda$ be any region. Define the
external boundary $\partial_+ M$ as a set of all sites $u\in
\bar{M}$ such that there is at least one projector $\Pi_a$
acting on both $u$ and $M$. Define the internal boundary as
$\partial_- M=\partial_+ \bar{M}$. Finally, define $\partial M
=
\partial_- M \cup \partial_+ M$.
\end{dfn}

\begin{center}
\begin{figure}[htb]
\centerline{
\mbox{
 \includegraphics[height=3cm]{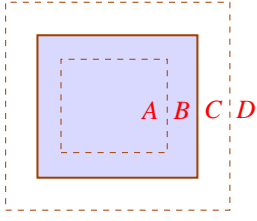}
 }}
\caption{Extending the correctability from a region $AB$ to a larger
region $ABC$. The Disentangling Lemma implies that any encoded state
$\rho$ can be represented as $\rho=U_{BC} (\eta_{AB} \otimes
\eta_{CD}) U_{BC}^\dag$, where $\eta_{AB}$ is a pure state
independent of $\rho$. It implies $\trace_C \eta_{CD}=\rho_D$ and
thus $\eta_{A} \otimes \rho_D =\calE(\rho)$, where $\calE$ is an
`error' erasing the region $BC$. If $BC$ is correctable, one must be able to reconstruct
$\rho$ starting from $\calE(\rho)$. Since $\eta_{AB}$ is known, it
means that one can reconstruct $\rho$ starting from $\rho_D$.
Therefore $ABC$ is correctable.}
 \label{fig:recursion}
\end{figure}
\end{center}
\begin{lemma}[\bf Disentangling]
\label{lemma:disent}
Let $M\subseteq \Lambda$ be any correctable region.
Suppose the external boundary $\partial_+M$ is also a correctable region.
 Then there exists a unitary operator
$U_{\partial M}$ acting only on the boundary  $\partial M$ such that
\be \label{disent} U_{\partial M} \Pi U_{\partial M}^\dag =
|\phi_M\ra\la \phi_M| \otimes \Pi_{\bar{M}}. \ee for some pure
state $|\phi_M\ra$ and some projector $\Pi_{\bar{M}}$.
\end{lemma}
It follows from Eq.~(\ref{disent}) that $U_{\partial M}$
disentangles any encoded state $|\psi\ra\in \calC$, that is,
$U_{\partial M}\, |\psi\ra =|\psi_{in}\ra \otimes |\psi_{out}\ra$
where $|\psi_{in}\ra  = |\phi_M\ra$ is the same for all encoded
states $|\psi\ra$.
In particular, any encoded state $\rho$ obeys the entanglement area law,
that is, $S(M)\le |\partial M|$ for any correctable region $M$. In the case of trivial codes ($k=0$)
the codespace is one-dimensional and thus any region is correctable. It reproduces
the entanglement area law proved for ground states of local Hamiltonians with commuting interactions in Ref.~\cite{Wolf07}.
Note that the boundary $\partial M_+$ is
correctable whenever the size of the boundary is smaller than the
distance $d$.  One can easily check that this is the case for all
applications of the lemma used above.

\omitted{Finally, it is worth mentioning that $\Pi_{\bar{M}}$
might not be representable as a product of geometrically local
projectors because the unitary operator $U_{\partial M}$ might not
be locality preserving. Similarly, the state $|\phi_M\ra$ might lack
a representation in terms of local commuting projectors.}


The proof of the Disentangling Lemma is based on the following
well-known result~\cite{BV03}.
\begin{prop}
\label{prop:2}
Consider a tripartite system $ABC$ and let
$\Pi=\Pi_{AB} \Pi_{BC} = \Pi_{BC}\Pi_{AB}$ be a product of two commuting projectors
acting on $AB$ and $BC$ respectively.
 Then the Hilbert space of $B$ can be decomposed as
\be
\calH_B = \bigoplus_x \calH_{B_x'} \otimes \calH_{B_x''}
\ee
such that
\be
\Pi=\bigoplus_x \Pi_{AB_x'} \otimes \Pi_{B_x''C}
\ee
for some projectors $\Pi_{AB_x'}$ and $\Pi_{B_x''C}$.
\end{prop}
Note that some of the projectors in the above decomposition might be zero.

\begin{proof}[\it Proof of the Disentangling Lemma]
Consider a partition $\Lambda=ABCD$, where
\[
A=M\backslash \partial_-M, \quad B=\partial_- M, \quad
C=\partial_+M, \quad D=\bar{M}\backslash \partial_+ M.
\]
By definition,  $M=AB$ and $\bar{M}=CD$, see
Fig.~\ref{fig:recursion}. Using Eq.~(\ref{PI}) one can represent
$\Pi$ as a product of commuting projectors acting on $MC$ and $CD$.
Then Proposition~\ref{prop:2} implies that the Hilbert space of $C$
can be decomposed as \be \label{decompos} \calH_C = \bigoplus_x
\calH_{C'_x} \otimes \calH_{C''_x} \ee such that \be \Pi=\bigoplus_x
\Pi_{MC'_x}^{(x)} \otimes \Pi_{C''_x D}^{(x)}, \ee where
$\Pi_{MC'_x}^{(x)}$ and $\Pi_{C''_x D}^{(x)}$ are projectors. Since
$C$ is correctable, the direct sum over $x$ contains exactly one
term --- otherwise it would be possible to distinguish some
orthogonal encoded states by measuring $x$ which can be done locally
in $C$, see Eq.~(\ref{decompos}). Thus one can subdivide $C$ into
two subsystems $C=C'C''$ such that \be \label{PIpart1} \Pi=U_C
\left( \Pi_{MC'} \otimes \Pi_{C''D}\right) U_C^\dag. \ee Using
Eq.~(\ref{PI}) again one can represent $\Pi$ as a product of
commuting projectors acting on $AB$ and $B\bar{M}$.  Applying
the same arguments as above one arrives at \be \label{PIpart2}
\Pi=U_B \left( \Pi_{AB'} \otimes \Pi_{B'' \bar{M}}\right)
U_B^\dag \ee where $B=B'B''$ is a partition of $B$ into two
subsystems and $\Pi_{AB'}$, $\Pi_{B'' \bar{M}}$ are some
projectors. Define a new projector \be \label{PIprime} \Pi'=U_B^\dag
U_C^\dag \Pi U_B U_C. \ee Combining
Eqs.~(\ref{PIpart1},\ref{PIpart2}) one concludes that $\Pi'$ has a
product structure with respect to the partition
$(AB')(B''C')(C''D)$, that is, \be \Pi' = \Pi_{AB'} \otimes
\Theta_{B''C'} \otimes \Pi_{C''D} \ee for some projector
$\Theta_{B''C'}$. The error correction condition Eq.~(\ref{EC}) for
$M$ implies that $\Pi_{AB'}$ must be one-dimensional, since
otherwise one would be able to find a pair of  orthogonal codestates
which can be distinguished by acting only on $M$. Thus \be
\label{PI'} \Pi'=|\phi_{AB'}\ra\la \phi_{AB'} |\otimes
\Theta_{B''C'} \otimes \Pi_{C''D} \ee for some pure state
$|\phi_{AB'}\ra$. As for the projector $\Theta_{B''C'}$, the error
correction condition Eq.~(\ref{EC}) for $M$  and $C$ (separately) implies that
$\Theta_{B''C'}$ can be regarded as a codespace of an error
correcting code that corrects all errors on $B''$ and all errors on
$C'$. The no-cloning principle implies that $\Theta_{B''C'}$ must be
one-dimensional,  that is, \be \label{PI''} \Pi' = |\phi_{AB'}\ra\la
\phi_{AB'} | \otimes |\phi_{B''C'}\ra\la \phi_{B''C'}| \otimes
\Pi_{C''D} \ee for some pure state $|\phi_{B''C'}\ra$. Thus the
desired unitary operator $U_{\partial M}$ can be chosen as \be
U_{\partial M} = W_{B''C'} U_B^\dag U_C^\dag \ee where $W_{B''C'}$
is an arbitrary unitary operator disentangling the state $
|\phi_{B''C'}\ra$.
\end{proof}

\omitted{
\begin{proof}[\bf Proof of Corollary~\ref{cor:recursion}]
Applying the Disentangling Lemma to the region $M=AB$ we conclude
that there exists a unitary operator $U_{BC}$ and a pure state
$\eta_{AB}$ such that for any encoded state $\rho$ one has \be
\label{disentBC} \rho=U_{BC}(\eta_{AB} \otimes \eta_{CD} )
U_{BC}^\dag, \ee where $\eta_{CD}$ is some (mixed) state depending
on $\rho$. Taking the partial trace of Eq.~(\ref{disentBC}) over
$ABC$ we conclude that $\trace_C \eta_{CD} =\rho_D$. Therefore \be
\eta_{A} \otimes \rho_D = \calE(\rho), \ee where we introduced an
`error' $\calE$ that takes the partial trace over $BC$.
If $BC$ is
correctable, there exists a recovery operation $\calR$ such that
$\calR\circ \calE (\rho) = \rho$ for any encoded state $\rho$.
Therefore \be \rho=\calR(\eta_{A} \otimes \rho_D). \ee Since
$\eta_{A}$ is a known state independent of $\rho$, it means that
one can reconstruct $\rho$ starting from $\rho_D$. Therefore $ABC$
is a correctable region.
\end{proof}}
Our final lemma asserts that the union of two correctable regions
$M_1$ and $M_2$ that are sufficiently far apart is also correctable.
Note that this statement would be obvious if the error correction
would amount to the ``syndrome measurement", that is, measuring
eigenvalues of the constraints $\Pi_a$ and guessing the error based
on the measured syndrome. Indeed, an error acting on a region $M_i$
creates non-trivial syndrome only in a small neighborhood of $M_i$,
so the error corrections at $M_1$ and $M_2$ do not interfere with
each other.  Unfortunately, this intuition does not lead to a formal
proof, so we need to use different arguments similar to the ones
used in the proof of Lemma~\ref{lemma:disent}.
\begin{lemma}
\label{lemma:union}
Let $M_1,M_2\subseteq \Lambda$ be any correctable regions such that
any projector $\Pi_a$ overlaps with at most one of $M_1,M_2$.
Suppose that
$\partial_+ M_1$ is  also correctable.
Then the region $M_1\cup M_2$ is correctable.
\end{lemma}
\begin{proof}
It suffices to prove that
\be
\label{M1M2}
\Pi \, O_{M_1} \otimes O_{M_2} \, \Pi \sim \Pi
\ee
for any operators $O_{M_1}, O_{M_2}$ acting on $M_1, M_2$ respectively.
Indeed, since any projector $\Pi_a$ overlaps with at most one of $M_1,M_2$
the regions $M_1 \cup \partial_+ M_1$ and $M_2$ are disjoint.
Let us apply the decomposition described by Eqs.~(\ref{PIprime}-\ref{PI''}) to the region $M_1$.
It yields
\be
\label{f}
\Pi \, O_{M_1} \otimes O_{M_2} \, \Pi =f(O_{M_1}) \, \Pi O_{M_2} \Pi,
\ee
where
\be
f(O_{M_1})=
\la \phi_{AB'} \otimes \phi_{B'' C'} |O_{AB'B''}| \phi_{AB'} \otimes \phi_{B'' C'}\ra
\ee
and $O_{AB'B''}=U_B^\dag \,  O_{M_1} U_B$.
Since $M_2$ is correctable, Eq.~(\ref{f}) implies Eq.~(\ref{M1M2}).
\end{proof}


\section{Acknowledgments}
We thank Frank Verstraete for useful discussions.
Part of this work was done while the authors were visiting
the Erwin Schr\"odinger International Institute for Mathematical Physics
at Vienna.
SB and BMT were partially  supported by the
DARPA QUEST program under contract number HR0011-09-C-0047.
DP is supported in part by NSERC and FQRNT.

\section{Appendix A}
In this section we prove the bound Eq.~(\ref{classical}) for 2D classical codes and
demonstrate that this bound might be tight by inspecting properties of 2D classical codes
associated with 1D cellular automatons.

In the classical case each site of the lattice $u\in \Lambda$ is
occupied by a classical variable $x_u$ that can take a constant
number of values. The codespace $\calC$ is a set of all assignments
$x=\{x_u\}_{u\in \Lambda}$ that obey  geometrically local
constraints $\Pi_1(x)=1,\ldots,\Pi_m(x)=1$. A code encodes $k$ bits
with the distance $d$ iff $|\calC|=2^k$ and any pair of distinct
codewords differ at $d$ or more sites. Consider a partition
$\Lambda=AB$, where $B=B_1\ldots B_m$ consists of square-shaped
blocks of size roughly $\sqrt{d}\times \sqrt{d}$ such that the
number of sites in any block $B_i$ is smaller than $d$, see
Fig.~\ref{fig:classical}. We assume that the separation between the
blocks in $B$ is of order $w$, so that any constraint $\Pi_a$
overlaps with at most one block $B_i$. Let $x,y\in \calC$ be any
pair of codewords such that $x|_A=y|_A$. We claim that $x=y$.
Indeed, suppose $x$ and $y$ differ at some block $B_i$. Then there
exists a codeword $z\in \calC$ that coincides with $x$ inside $B_i$
and coincides with $y$ in the complement of $B_i$. It means that $z$
and $y$ are distinct codewords that differ at less than $d$ sites
which is a contradiction.

Let $\rho$ be the uniform distribution on $\calC$. We have $S(B|A)=0$ since
there is a unique way to extend a codeword from $A$ to $B$. Therefore
\[
k=S(\rho) =S(A) + S(B|A)=S(A)\le |A| \sim \frac{n}{d^{1/2}}.
\]
It proves Eq.~(\ref{classical}).
\begin{center}
\begin{figure}[htb]
\centerline{ \includegraphics[height=3cm]{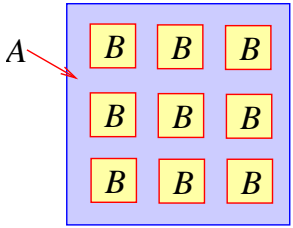}}
\caption{The partition $\Lambda=AB$.}
 \label{fig:classical}
\end{figure}
\end{center}
In the rest of the section we describe a family of 2D linear codes associated with 1D cellular automatons (CA) that
 are quite close to saturating the bound Eq.~(\ref{classical}). To the best of our knowledge the idea of CA-based codes was originally introduced
in Ref.~\cite{Chowdhury94}. A very similar construction  has also been used in Ref.~\cite{Moore99}
as an exactly solvable model of a 2D spin glass.

Let us start from considering  a semi-infinite lattice $\Lambda=\ZZ \times [0,L-1]$.   Let $x_i^t\in \{0,1\}$ be a classical bit
living at a site $(i,t)\in \Lambda$.
We shall refer to the coordinates $i$ and $t$ as {\em space} and {\em time} respectively.
Let $\{0,1\}^\Lambda$ be the set of all bit assignments
$\{x_i^t\}_{(i,t)\in \Lambda}$ with a finite Hamming weight.
Define a code
\bea
\label{CA1}
\calC_{\infty}^{L}&=&\{ x\in \{0,1\}^{\Lambda} \, : \, x^{t+1}_i = x^t_{i-1} \oplus x^t_{i+1} \nn \\
&& \forall i\in \ZZ, \quad \forall t\in [0,L-2]\}
\eea
Note that all constraints are linear and involve a triple of bits located close to each other.
 Clearly there is a one-to-one correspondence between
the codewords of $\calC_{\infty}^{L}$ and computational histories of a 1D linear cellular automaton (CA)  with
transition rules $x_i \to x_{i-1} \oplus x_{i+1}$. Accordingly, any codeword  $x\in \calC_{\infty}^{L}$ is uniquely
determined by the restriction of $x$ onto the first row of the lattice which determines the initial conditions
for the CA at $t=0$. It means that the code $\calC_{\infty}^{L}$ has $1$ encoded bit per unit of length along the space axis.
Since the code $\calC_{\infty}^{L}$ is linear, its distance $d$ is just the minimum Hamming weight of a non-zero
codeword $x\in \calC_{\infty}^{L}$.

\begin{lemma}
\label{lemma:CA1} Let $d(p)$ be the distance of the code $\calC_{\infty}^{L}$ defined on
a lattice of height $L$, where $L=2^p$ for some integer
$p$. Then $d(p)=3^p$.
\end{lemma}
\begin{proof}
Clearly $\calC_{\infty}^{L}$ consists of two independent codes defined on the even and odd sublattices of $\Lambda$.
Let $\Lambda_0$ be the even sublattice, i.e., a set of all sites $(i,t)\in \Lambda$ such that $i+t$ is even.
It suffices to bound the distance of the code $\calC_{\infty}^{L}$ restricted to $\Lambda_0$.
Consider a partition $\Lambda_0=ABCD$ where
\bea
A&=&\{ (i,t) \, : \, i+t=0 \bmod{4}, \quad t=0 \bmod{2}\} \nn \\
B&=&\{ (i,t) \, : \, i+t=2 \bmod{4}, \quad t=0 \bmod{2}\} \nn \\
C&=&\{ (i,t) \, : \, i+t=0 \bmod{4}, \quad t=1 \bmod{2}\} \nn \\
D&=&\{ (i,t) \, : \, i+t =2 \bmod{4}, \quad t=1 \bmod{2}\} \nn
\eea
Note that each of the above sublattices  is isomorphic to the original lattice
$\Lambda_0$ of height $2^{p-1}$.
Using the transition rules $x^{t+1}_i=x^t_{i-1}\oplus x^t_{i+1}$
one easily gets $x^{t+2}_i = x^t_{i-2} \oplus x^t_{i+2}$, that is, the code $\calC_{\infty}^{L}$
reproduces itself on each of the sublattices $A,B,C,D$.  We conclude that
\be
d(p)\ge \Gamma \, d(p-1)
\ee
where $\Gamma$ is the minimum number of sublattices $A,B,C,D$ that can be occupied by a non-zero codeword.
Simple combinatorial analysis shows that $\Gamma\ge 3$, that is, $d(p)\ge 3^p$.

To get the matching upper bound on $d(p)$ consider a codeword $x\in \calC_{\infty}^L$
generated starting from a state with a single active cell, i.e., a codeword corresponding to the initial
conditions $x^0_i=\delta_{i,1}$.
One can easily check that the support of $x$ is a discrete version of the Sierpinski triangle fractal which has Hamming weight
$3^p$. Thus $d(p)=3^p$.
\end{proof}
Consider now a finite lattice $\Lambda=\ZZ_L \times [0,L-1]$ with periodic boundary conditions
along the space axis and open boundary conditions along the time axis.
Define a finite version of the  code $\calC_{\infty}^L$ by the  constraints
\be
\label{CA1finite}
x^{t+1}_i = x^t_{i-1} \oplus x^t_{i+1}, \quad  x^0_{0}=0,
\ee
which must hold for all $i\in \ZZ_L$ and for all $t\in [0,L-2]$.
Let us denote the corresponding code $\calC^L_L$.
We shall restrict ourselves only to odd values of $L$. One can easily check that for odd $L$
the transition rule $x_i\to x_{i-1} \oplus x_{i+1}$ is essentially reversible: a pair of distinct
initial states $\{x^0_i\}_{i\in \ZZ_L}$ and $\{y^0_i\}_{i\in \ZZ_L}$ can evolve into the same state after a finite number of steps iff  $x^0_i=y^0_i\oplus 1$ for all $i\in \ZZ_L$.
The additional constraint $x^0_0=0$ thus guarantees that distinct codewords of $\calC_{L}^{L}$ have distinct
restrictions on every time slice of the lattice.
By construction,
the modified code $\calC_{L}^{L}$ encodes $k=L-1$ bits into $n=L^2$ bits.
We have computed the distance $d$ of the code $\calC_{L}^{L}$ numerically using  the exhaustive search optimization for odd
values of $L$ in the interval $5\le L\le 23$, see Fig.~\ref{fig:distance_plot}.
It was checked that  for all considered values of $L$ one has $d=d'$, where $d'$
is the Hamming weight of a codeword generated starting from a state  with  a single active cell,
that is, with the initial conditions $x^0_i=\delta_{i,1}$.
Since we have also shown that $d=d'$ for the semi-infinite lattice, see the proof of Lemma~\ref{lemma:CA1},
it is natural to conjecture that $d=d'$ for all odd values of $L$.
Computing $d'$ numerically for lattice sizes up to $L\sim 10^4$ we have found
$d'\sim L^{1.584}$ which agrees perfectly with the scaling
\be
\label{CA1distance}
d\sim L^{\log_2 3} \approx L^{1.585} \approx n^{0.793}
\ee
that was derived in Lemma~\ref{lemma:CA1} for the semi-infinite lattice.
Summarizing, the code $\calC_{L}^{L}$ encodes $k=\sqrt{n}-1$ bits into $n$ bits with the distance
$d\approx n^{0.793}$. Note that $k\sqrt{d} \sim n^{0.897}$ which is quite close to saturating the bound Eq.~(\ref{classical}).

\begin{center}
\begin{figure}[htb]
\centerline{
\mbox{
 \includegraphics[height=5cm]{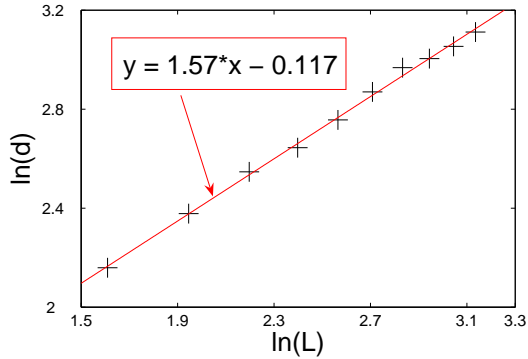}
 }}
 \caption{The distance of the code $\calC_{L}^{L}$ defined on a lattice
 $\ZZ_L\times [0,L-1]$ computed numerically using the exhaustive search for
 odd $L$ in the interval $5\le L\le 23$.}
 \label{fig:distance_plot}
\end{figure}
\end{center}


\end{document}